\documentclass[11pt]{article}%
\usepackage{amsfonts}
\usepackage{amsmath}
\usepackage{amssymb}
\usepackage{graphicx}%
\providecommand{\U}[1]{\protect\rule{.1in}{.1in}}
\newtheorem{theorem}{Theorem}
\newtheorem{acknowledgement}[theorem]{Acknowledgement}

\newtheorem{definition}[theorem]{Definition}

\newtheorem{lemma}[theorem]{Lemma}

\newtheorem{remark}[theorem]{Remark}

\newenvironment{proof}[1][Proof]{\noindent\textbf{#1.} }{\ \rule{0.5em}{0.5em}}
\begin{document}

\title{Symplectic Radon Transform and the Metaplectic Representation}
\author{Maurice A. de Gosson\\University of Vienna\\Faculty of Mathematics (NuHAG)\\Oskar-Morgenstern-Platz 1, 1090 Vienna (AUSTRIA)\\maurice.de.gosson@univie.ac.at}
\maketitle
\tableofcontents

\begin{abstract}
We study the symplectic Radon transform from the point of view of the
metaplectic representation of the symplectic group and its action on the
Lagrangian Grassmannian. We give rigorous proofs in the general setting of
multi-dimensional quantum systems. We interpret the inverse Radon transform as
a \textquotedblleft demarginalization process\textquotedblright\ for the
Wigner distribution. This work completes, by giving complete proofs, a prvious Note.

\end{abstract}

\section{Introduction}

The idea of using what is today called the \textquotedblleft Radon
transform\textquotedblright\ to reconstruct a function from partial data goes
back to the 1917 work \cite{Radon} by the Austrian mathematician Johann Radon.
While Radon originally only considered the transform of two or three
dimensional systems (in which case it is called the \textquotedblleft X-ray
transform\textquotedblright), the theory has since then be generalized to
arbitrary Euclidean spaces. The applications of the Radon transform to quantum
mechanics and optics have been developing rapidly these last years (see for
instance the seminal paper \cite{VR} by Vogel and Risken). In
\cite{Asorey,manko,Facchi,Ibort,Ibortbis,manko2} the authors study what they
call the \textquotedblleft symplectic Radon transform\textquotedblright\ of a
mixed quantum state; in a recent \cite{Ibort2021} some of these results are
extended to the framework of $C^{\ast}$ algebras. The aim of the present paper
is to give a simple rigorous approach to the theory of the quantum Radon
transform in several degrees of freedom. For this we will use systematically
the the theory of the metaplectic group as developed in our previous works
\cite{Birk}, together with the elementary theory of Lagrangian subspaces of
the standard symplectic space. Our main observation is the following:
integration of the Wigner transform $W\psi(x,p)$ (of a function, or of a
state) along the $x$ and $p$ coordinate planes give the correct probability
distributions $|\psi(x)|^{2}$ and $|\widehat{\psi}(p)|^{2}$ in position and
momentum space. However, these are not sufficient to reconstruct the state
$\psi$ (this is an aspect of the \textquotedblleft Pauli
problem\textquotedblright, see \S \ref{sec3}). However, applying a metaplectic
transform $\widehat{U}$ associated with a symplectic rotation $U$ to the state
$\psi$ transforms the Wigner transform following the rule $W(\widehat{U}%
\psi)(x,p)=W\psi(U^{-1}(x,p))$ (this is the \textquotedblleft symplectic
covariance property\textquotedblright\ of the -Wigner transform, well-known in
harmonic analysis, especially in the Weyl--Wigner --Moyal approach to quantum
mechanics). This use of metaplectic transforms leads, by calculating the
marginals of $W(\widehat{U}\psi)(x,p)$ along the $x$ and $p$ coordinate planes
to infinitely many probability distributions, and these allow the
reconstruction of the state. More precisely, we will redefine the Radon
transform as being given by the formula%
\[
R_{\psi}(X,A,B)=\det\Lambda^{-1}|\widehat{U}_{A,B}\psi(\Lambda^{-1}X)|^{2}%
\]
where $X\in\mathbb{R}^{n}$ and $\widehat{U}_{A,B}$ is the metaplectic operator
associated with the symplectic rotation
\[
U_{AB}=%
\begin{pmatrix}
\Lambda^{-1}A & \Lambda^{-1}B\\
-\Lambda^{-1}B & \Lambda^{-1}A
\end{pmatrix}
\]
where $A,B$ are square matrices with $\operatorname*{rank}(A,B)=n$ and
$\Lambda=(A^{T}A+B^{T}B)^{1/2}$. We thereafter prove (under suitable
conditions on $\psi$) the inversion formula
\[
W\psi(x,p)=(2\pi\hbar)^{-2n^{2}}\int R_{\psi}(X,A,B)e^{\frac{i}{\hslash
}(X-Ax-Bp)}dXdAdB
\]
which reduces for $n=1$ to the usual inversion formula found in the literature
(the case $n=1$ is discussed in \S \ref{sec1} to motivate the results in the
general case). As an illustration, we apply our constructions to the
generalized Gaussian states
\begin{equation}
\psi_{V,W}(x)=\left(  \tfrac{1}{\pi\hbar}\right)  ^{n/4}(\det V)^{1/4}%
e^{-\tfrac{1}{2\hbar}(V+iW)x^{2}}~.
\end{equation}
in \S , which gives us the opportunity to shortly discuss the Pauli problem
for Gaussians, thus generalizing results in \cite{MC}.

\section{The Case $n=1$\label{sec1}}

Let $\widehat{\rho}$ be a mixed quantum state with one degree of freedom:
$\widehat{\rho}$ is a positive-semidefinite trace class operator on
$L^{2}(\mathbb{R})$ with trave $\operatorname*{Tr}(\widehat{\rho})=1$. In view
of the spectral theorem there exists a sequence $(\psi_{j})$ with $\psi_{j}\in
L^{2}(\mathbb{R})$ and a sequence of non-negative numbers $(\lambda_{j})$ with
$\sum_{j}\lambda_{j}=1$ such that $\widehat{\rho}=\sum_{j}\lambda_{j}|\psi
_{j}\rangle\langle\psi_{j}|$ where $|\psi_{j}\rangle\langle\psi_{j}|$ is the
orthogonal projection of $L^{2}(\mathbb{R})$ onto the ray $\mathbb{C}%
\psi_{j\text{.}}$. By definition \cite{QUANTA,QHA}, the Wigner distribution of
$\widehat{\rho}$ is the convex sum%
\[
\rho=\sum_{j}\lambda_{j}W\psi_{j}%
\]
where $W\psi_{j}$ is the usual Wigner transform of $\psi_{j}$, defined for
$\psi\in L^{2}(\mathbb{R})$ by
\[
W\psi(x,p)=\frac{1}{2\pi\hbar}\int_{-\infty}^{\infty}e^{-\frac{i}{\hbar}%
py}\psi(x+\tfrac{1}{2}y)\psi^{\ast}(x-\tfrac{1}{2}y)dy.
\]
Recall \cite{Wigner,QUANTA,QHA} that the marginal properties%
\begin{equation}
\int_{-\infty}^{\infty}W\psi(x,p)dp=|\psi(x)|^{2}\text{ \ , \ }\int_{-\infty
}^{\infty}W\psi(x,p)dx=|\widehat{\psi}(p)|^{2} \label{marginal}%
\end{equation}
make sense provided that $\psi$ and its Fourier transform are, in addition to
being square integrable, absolutely integrable: $\psi,\widehat{\psi}\in
L^{1}(\mathbb{R})$.

In most texts studying the tomographic picture of quantum mechanics the
symplectic Radon transform of the quantum state $\widehat{\rho}$ is defined by
the integral
\begin{equation}
R_{\widehat{\rho}}(X,a,b)=\int\rho(x,p)\delta(X-ax-bp)dpdx \label{1}%
\end{equation}
where $a$ and $b$ are real numbers, and it is claimed that the following
essential reconstruction formula holds
\begin{equation}
\rho\psi(x,p)=\frac{1}{2\pi\hbar}\int R_{\widehat{\rho}}(X,a,b)e^{\frac
{i}{\hslash}(X-ax-bp)}dXdadb. \label{3}%
\end{equation}

The following result is simultaneously a rigorous restatement and a
justification of these formulas. It will be extended to the case of quantum
states with an arbitrary number $n$ of freedom in the forthcoming sections.
Among other things, we see that the inverse Radon transform can be viewed as a
\textquotedblleft demarginalization process\textquotedblright\ \cite{PNAS} for
the Wigner distribution.

We will use the following notation: we set $U_{a,b}=%
\begin{pmatrix}
a/\lambda & b/\lambda\\
-b/\lambda & a/\lambda
\end{pmatrix}
$ where $\lambda=\sqrt{a^{2}+b^{2}}$; clearly $U_{a,b}$ is a rotation in the
$x,p$ plane.

\begin{theorem}
\label{Thm1}Let $\widehat{\rho}$ be a pure quantum state: $\rho=2\pi\hbar
W\psi$ for some $\psi\in L^{2}(\mathbb{R})$. We assume that in addition
$\psi,\widehat{\psi}\in L^{1}(\mathbb{R})$, which ensures that the marginal
properties (i) The Radon transform $R_{\widehat{\rho}}(X,a,b)$ is given by the
formula
\begin{equation}
R_{\widehat{\rho}}(X,a,b)=\lambda^{-1}|\widehat{U}_{a,b}\psi(\lambda
^{-1}X)|^{2} \label{Radon1}%
\end{equation}
where $\widehat{U}_{a,b}\in\operatorname*{Mp}(n)$ is anyone of the two
metaplectic operators covering the rotation $U_{a,b}$. (ii) The inverse Radon
transform is given by the formula:%
\begin{equation}
W\psi(x,p)=\frac{1}{2\pi\hbar}\int R_{\widehat{\rho}}(X,a,b)e^{\frac
{i}{\hslash}(X-ax-bp)}dXdadb. \label{inverse}%
\end{equation}
(iii)The Radon transform of $\psi$is given by the line integral
\begin{equation}
R_{\widehat{\rho}}(X,a,b)=\int_{-}^{\infty}W\psi(z(t))|\dot{z}(t)|dt
\label{RadonGeom}%
\end{equation}
where $t\longmapsto z(t)$ is a parametrization of the straight line
$\ell_{a,b}^{X}$ in $\mathbb{R}^{2}$ with equation $ax+bp=X$.
\end{theorem}

\begin{proof}
It is sufficient to assume that $\widehat{\rho}$ is a pure state, that is
$\rho=2\pi\hbar W\psi$ for some $\psi$. (i) Let us make the change of
variables
\begin{equation}%
\begin{pmatrix}
u\\
v
\end{pmatrix}
=%
\begin{pmatrix}
a/\lambda & b/\lambda\\
-b/\lambda & a/\lambda
\end{pmatrix}%
\begin{pmatrix}
x\\
p
\end{pmatrix}
\label{unitary}%
\end{equation}
in the integral (\ref{1}).This leads to the expression%
\begin{equation}
R_{\widehat{\rho}}(X,a,b)=\iint W\psi(U_{a,b}^{-1}(u,v))\delta(X-\lambda
u)dudv. \label{dirac}%
\end{equation}
Since $\delta(X-\lambda u)=\lambda^{-1}\delta(\lambda^{-1}X-u)$ this can be
rewritten
\begin{equation}
R_{\widehat{\rho}}(X,a,b)=\lambda^{-1}\iint W\psi(\lambda^{-1}%
(au-bv,bu+av))\delta(\lambda^{-1}X-u)dudv \label{diracbis}%
\end{equation}
In view of the symplectic covariance property \cite{Birk,Wigner,Littlejohn} of
the Wigner transform we have%
\begin{equation}
W\psi(U_{a,b}^{-1}(u,v))=W(\widehat{U}_{a,b}\psi)(u,v) \label{sympco}%
\end{equation}
where $\widehat{U}_{a,b}$ is anyone of the two metaplectic operators (see the
Appendix A) covering $U$ and hence (\ref{diracbis}) yields
\begin{align*}
R_{\widehat{\rho}}(X,a,b)  &  =\lambda^{-1}\iint W(\widehat{U}_{a,b}%
\psi)(\lambda^{-1}X,v)\delta(\lambda^{-1}X-u)dudv\\
&  =\lambda^{-1}\int_{-\infty}^{\infty}W(\widehat{U}_{a,b}\psi)(\lambda
^{-1}X,v)dv
\end{align*}
hence formula (\ref{Radon1}) using the marginal properties (\ref{marginal}).
(ii) Let us denote $A$ the right-hand side of the equality (\ref{inverse}).
Using the first marginal property (\ref{marginal}) we have
\begin{align*}
A  &  =\lambda^{-1}\frac{1}{2\pi\hbar}\int_{\mathbb{R}^{3}}|\widehat{U}%
_{a,b}\psi(\lambda^{-1}X)|^{2}e^{\frac{i}{\hslash}(X-ax-bp)}dXdadb\\
&  =\lambda^{-1}\frac{1}{2\pi\hbar}\int_{\mathbb{R}^{4}}W(\widehat{U}%
_{a,b}\psi)(\lambda^{-1}X,P)e^{\frac{i}{\hslash}(X-ax-bp)}dXdPdadb.
\end{align*}
Replacing $X$ with $\lambda X$ and using the symplectic covariance property
(\ref{sympco}) we get
\begin{align*}
A  &  =\frac{1}{2\pi\hbar}\int_{\mathbb{R}^{4}}W(\widehat{U}_{a,b}%
\psi)(X,P)e^{\frac{i}{\hslash}(\lambda X-ax-bp)}dXdPdadb\\
&  =\frac{1}{2\pi\hbar}\int_{\mathbb{R}^{4}}W\psi(U_{a,b}^{-1}(X,P))e^{\frac
{i}{\hslash}(\lambda X-ax-bp)}dXdPdadb\\
&  =\frac{1}{2\pi\hbar}\int_{\mathbb{R}^{4}}W\psi((a/\lambda)X-(b/\lambda
)P,(b/\lambda)X+(a/\lambda)P))e^{\frac{i}{\hslash}(\lambda X-ax-bp)}dXdPdadb.
\end{align*}
Setting $Y=(a/\lambda)X-(b/\lambda)P$ and $Z=(b/\lambda)X+(a/\lambda)P$ (and
hence $\lambda X=aY+bZ$) we have $dXdP=dYdZ$ so that
\[
A=\frac{1}{2\pi\hbar}\int_{\mathbb{R}^{4}}W\psi(Y,Z))e^{\frac{i}{\hslash
}(a(Y-x)+b(Z-p))}dYdZdadb.
\]
In view of the Fourier inversion formula, written formally as
\[
\iint e^{\frac{i}{\hslash}(a(Y-x)+b(Z-p))}dadb=2\pi\hbar\delta(Y-x,Z-p)
\]
we thus have
\[
A=\iint W\psi(x,p)\delta(Y-x,Z-p)dYdZ=W\psi(x,p)
\]
which was to be proven. (iii) It is sufficient to show that (\ref{RadonGeom})
holds for one parametrization. In view of formula (\ref{dirac}) we have%
\begin{align}
R_{\widehat{\rho}}(X,a,b)  &  =\lambda^{-1}\iint W\psi(U_{a,b}^{-1}%
(u,v))\delta(\lambda^{1}X-u)dudv\\
&  =\lambda^{-1}\int_{-\infty}^{\infty}\left(  \int_{-\infty}^{\infty}%
W\psi(U_{a,b}^{-1}(u,v))\delta(\lambda^{1}X-u)du\right)  dv\\
&  =\int_{-\infty}^{\infty}W\psi(U_{a,b}^{-1}(\lambda^{-1}X,v))dv
\end{align}
that is, since $U_{a,b}^{-1}=%
\begin{pmatrix}
a/\lambda & -b/\lambda\\
b/\lambda & a/\lambda
\end{pmatrix}
$, and replacing $v$ with $t$,
\begin{equation}
R_{\widehat{\rho}}(X,a,b)=\int_{-\infty}^{\infty}W\psi(a\lambda^{-2}%
X-v\lambda^{-1}t,v\lambda^{-2}X+a\lambda^{-1}t)dt. \label{parabis}%
\end{equation}
Set now $x(t)=a\lambda^{-2}X-v\lambda^{-1}t$ and $p(t)=v\lambda^{-2}%
X+a\lambda^{-1}t$. Then $ax(t)+bp(t)=X$ and $\dot{x}(t)^{2}+\dot{p}(t)^{2}=1$
hence (\ref{parabis}) implies (\ref{RadonGeom}).
\end{proof}

\begin{remark}
The first part of the theorem above uses the physicist's definition (\ref{1})
and can thus be taken as a mathematically correct redefinition of the Radon
transform. We will exploit this fact in next section.
\end{remark}

\section{The Multivariate Case\label{sec2}}

For $n\geq1$ we consider $\mathbb{R}^{2n}\equiv T^{\ast}\mathbb{R}^{n}$
equipped with its standard symplectic structure, defined by
\[
\sigma(z,z^{\prime})=Jz\cdot z^{\prime}\text{ \ , \ }J=%
\begin{pmatrix}
0 & I\\
-I & 0
\end{pmatrix}
.
\]
We denote by $\operatorname*{Sp}(n)$ the symplectic group of $(\mathbb{R}%
^{2n},\sigma)$ and by $Mp(n)$ its unitary representation of its double cover
(the metaplectic group; see Appendix A).

\subsection{Definitions}

$A,B$ be two real square $n\times n$ matrices with such that $A^{T}B=BA^{T}$
and $\operatorname*{rank}(A,B)=n$. Setting
\[
M_{AB}=%
\begin{pmatrix}
A & B\\
-B & A
\end{pmatrix}
\text{ \ },\text{ \ }\Lambda=(A^{T}A+B^{T}B)^{1/2}%
\]
and noting that $A^{T}A+B^{T}B$ is invertible, we have the factorization
\[
M_{AB}=%
\begin{pmatrix}
\Lambda & 0\\
0 & \Lambda
\end{pmatrix}%
\begin{pmatrix}
\Lambda^{-1}A & \Lambda^{-1}B\\
-\Lambda^{-1}B & \Lambda^{-1}A
\end{pmatrix}
\]
where
\begin{equation}
U_{A,B}=%
\begin{pmatrix}
\Lambda^{-1}A & \Lambda^{-1}B\\
-\Lambda^{-1}B & \Lambda^{-1}A
\end{pmatrix}
\in U(n) \label{UAB}%
\end{equation}
is a symplectic rotation. Note that its inverse is
\begin{equation}
U_{A,B}^{-1}=%
\begin{pmatrix}
A^{T}\Lambda^{-1} & -B^{T}\Lambda^{-1}\\
-B^{T}\Lambda^{-1} & A^{T}\Lambda^{-1}%
\end{pmatrix}
. \label{invUAB}%
\end{equation}

\begin{definition}
The symplectic Radon transform of $\psi$ is the transformation
\[
R_{\psi}(\cdot,A,B):L^{2}(\mathbb{R}^{n})\longrightarrow L^{1}(\mathbb{R}%
^{n})
\]
defined, for $\psi\in L^{2}(\mathbb{R}^{n})$ by%
\begin{equation}
R_{\psi}(X,A,B)=\det\Lambda^{-1}|\widehat{U}_{A,B}\psi(\Lambda^{-1}X)|^{2}
\label{RadonDefine}%
\end{equation}
where $\widehat{U}_{A,B}:L^{2}(\mathbb{R}^{n})\longrightarrow L^{2}%
(\mathbb{R}^{n})$ is any of the two elements of $\operatorname*{Mp}(n)$
covering the symplectic rotation $U_{A,B}$.
\end{definition}

When $\det B\neq0$ the metaplectic operator $\widehat{U}_{A,B}$ is defined by
\begin{align}
\widehat{U}_{A,B}\psi(x)  &  =\left(  \tfrac{1}{2\pi\hbar}\right)
^{n/2}i^{m-n/2}\sqrt{|\det B^{-1}|}\int_{\mathbb{R}^{n}}e^{\frac{i}{\hbar
}\mathcal{A}(x,x^{\prime})}\psi(x^{\prime})dx^{\prime}\\
\mathcal{A}(x,x^{\prime})  &  =\frac{1}{2}AB^{-1}x\cdot x-B^{-1}x\cdot
x^{\prime}+\frac{1}{2}B^{-1}Ax^{\prime}\cdot x^{\prime}%
\end{align}
after one has made a choice of the integer $m$ modulo $4$.

Recall that in the multidimensional case the Wigner transform $W\psi$ of
$\psi\in L^{2}(\mathbb{R}^{n})$ is given by the integral%
\[
W\psi(x,p)=\left(  \frac{1}{2\pi\hbar}\right)  ^{n}\int_{\mathbb{R}^{n}%
}e^{-\frac{i}{\hbar}py}\psi(x+\tfrac{1}{2}y)\psi^{\ast}(x-\tfrac{1}{2}y)dy
\]
and that the marginal properties generalizing (\ref{marginal}) hold
\begin{equation}
\int_{\mathbb{R}^{n}}W\psi(x,p)dp=|\psi(x)|^{2}\text{ \ },\text{ \ }%
\int_{\mathbb{R}^{n}}W\psi(x,p)dx=|\widehat{\psi}(p)|^{2} \label{marginaln}%
\end{equation}
for $\psi,\widehat{\psi}\in L^{1}(\mathbb{R}^{n})$; the Fourier transform
$\widehat{\psi}=F\psi$ is here given by
\begin{equation}
F\psi(p)=\left(  \frac{1}{2\pi\hbar}\right)  ^{n/2}\int_{\mathbb{R}^{n}%
}e^{-\frac{i}{\hbar}p\cdot x}\psi(x)dx. \label{Fourier}%
\end{equation}

\subsection{The Radon inversion formula}

We have the following straightforward generalization of the inversion result
(ii) in Theorem \ref{Thm1}:

\begin{theorem}
Viewing $A$ and $B$ as elements of $\mathbb{R}^{n^{2}}$ we have
\begin{equation}
W\psi(x,p)=(2\pi\hbar)^{-2n^{2}}\int_{\mathbb{R}^{n(2n+1)}}R_{\psi
}(X,A,B)e^{\frac{i}{\hslash}(X-Ax-Bp)}dXdAdB. \label{inversen}%
\end{equation}

\end{theorem}

\begin{proof}
It goes exactly as the proof of Theorem \ref{Thm1} (ii). Let us denote by $W$
the integral in the right-hand side of (\ref{inversen}). We have, using the
first marginal condition (\ref{marginaln}),
\begin{align*}
W  &  =\det\Lambda^{-1}\int|\widehat{U}_{A,B}\psi(\Lambda^{-1}X)|^{2}%
e^{\frac{i}{\hslash}(X-Ax-Bp)}dXdAdB\\
&  =\det\Lambda^{-1}\int W(\widehat{U}_{A,B}\psi)(\Lambda^{-1}X,P)e^{\frac
{i}{\hslash}(X-Ax-Bp)}dPdXdAdB\\
&  =\int W(\psi)(U_{A,B}^{-1}(X,P))e^{\frac{i}{\hslash}(\Lambda X-Ax-Bp)}%
dPdXdAdB\\
&  =\int W(\psi)(U_{A,B}^{-1}(X,P))e^{\frac{i}{\hslash}(\Lambda X-Ax-Bp)}%
dPdXdAdB,
\end{align*}
that is, by formula (\ref{invUAB}):%
\begin{multline*}
W=\int W(\psi)(A^{T}\Lambda^{-1}X-B^{T}\Lambda^{-1}P,B^{T}\Lambda^{-1}%
X+A^{T}\Lambda^{-1}P))\\
\times e^{\frac{i}{\hslash}(\Lambda X-Ax-Bp)}dPdXdAdB.
\end{multline*}
Setting $Y=A^{T}\Lambda^{-1}X-B^{T}\Lambda^{-1}P$ and $Z=B^{T}\Lambda
^{-1}X+A^{T}\Lambda^{-1}P)$ we have $dYdZ=dPdX$ and hence
\[
W=\int W\psi(Y,Z)e^{\frac{i}{\hslash}(A(Y-x)+B(Z-p))}dYdZdAdB.
\]
Integration of the exponential with respect to the variables $A$ and $B$
yields%
\[
\int_{\mathbb{R}^{2n^{2}}}e^{\frac{i}{\hslash}(A(Y-x)+B(Z-p))}dAdB=(2\pi
\hbar)^{2n^{2}}%
\]
and hence
\begin{align*}
W  &  =(2\pi\hbar)^{2n^{2}}\int W\psi(Y,Z)\delta(Y-x,Z-p)|dYdZ\\
&  =(2\pi\hbar)^{2n^{2}}W\psi(x,p)
\end{align*}
which was to be proven.
\end{proof}

\subsection{Interpretation as generalized marginals}

Let us return to the definition (\ref{RadonDefine}) of the Radon transform:%
\[
R_{\psi}(X,A,B)=\det\Lambda^{-1}|\widehat{U}_{A,B}\psi(\Lambda^{-1}X)|^{2}.
\]
If we choose $A=I$ and $B=0$ this reduces to the formula%
\[
R_{\psi}(X,I,0)=|\psi(X)|^{2}=\int_{\mathbb{R}^{n}}W\psi(X,P)dP;
\]
similarly, if $A=0$ and $B=I$ we get%
\[
R_{\psi}(X,A,B)=|F\psi(X)|^{2}=\int_{\mathbb{R}^{n}}W\psi(P,X)dX.
\]
definition (\ref{RadonDefine}) reduces to the formula%
\[
R_{\psi}(X,A,B)=|\widehat{U}_{A,B}\psi(X)|^{2}%
\]
showing that the Radon transform is essentially a margin property for the
\textquotedblleft rotated\textquotedblright\ Wigner transform of $\psi$. In
fact, we can view $R_{\psi}(X,A,B)$ as the surface integral of the Wigner
transform on the (affine) Lagrangian plane (see Appendix B)
\[
\ell_{A,B}^{X}=\{(x,p)\in\mathbb{R}^{2n}:Ax+Bp=X\}.
\]

In view of the marginal properties and the symplectic covariance of the Wigner
transform, followed by the change of variables\ $P\longmapsto\Lambda^{-1}P$,
we have
\begin{align*}
|\widehat{U}_{A,B}\psi(\Lambda^{-1}X)|^{2}  &  =\int_{\mathbb{R}^{n}%
}W(\widehat{U}_{A,B}\psi)(\Lambda^{-1}X,P)dP\\
&  =\int_{\mathbb{R}^{n}}W\psi(U_{A,B}^{-1}(\Lambda^{-1}X,P))dP\\
&  =\det\Lambda\int_{\mathbb{R}^{n}}W\psi(U_{A,B}^{-1}(\Lambda^{-1}%
X,\Lambda^{-1}P))dP,
\end{align*}
that is, explicitly,
\begin{multline*}
|\widehat{U}_{A,B}\psi(\Lambda^{-1}X)|^{2}=\det\Lambda\\
\times\int_{\mathbb{R}^{n}}W\psi(A^{T}\Lambda^{-1}X-B^{T}\Lambda^{-1}%
P,B^{T}\Lambda^{-1}X+A^{T}\Lambda^{-1}P)dP.
\end{multline*}
Using the multi-parametrization
\begin{align*}
X^{\prime}(P)  &  =A^{T}\Lambda^{-1}X-B^{T}\Lambda^{-1}P\\
P^{\prime}(P)  &  =B^{T}\Lambda^{-1}X+A^{T}\Lambda^{-1}P
\end{align*}
we have%
\begin{equation}
R_{\psi}(X,A,B)=\int_{\mathbb{R}^{n}}W\psi(X^{\prime}(P),P^{\prime}(P)))dP
\label{geom}%
\end{equation}
we have $AX^{\prime}(P)+BP^{\prime}(P)=X$ so we can interpret the formula
above as a surface integral%
\[
|\widehat{U}_{A,B}\psi(\Lambda^{-1}X)|^{2}=\int_{\ell_{A,B}^{X}}W\psi
(Z)d\mu(Z)
\]
where $d\mu(Z)$ is the Lebesgue measure on $\ell_{A,B}^{X}$.

\section{Radon Transform of Generalized Gaussians\label{sec3}}

\subsection{Generalized Gaussians}

By generalized (centered) Gaussian we mean a function
\begin{equation}
\psi_{V,W}(x)=\left(  \tfrac{1}{\pi\hbar}\right)  ^{n/4}(\det V)^{1/4}%
e^{-\tfrac{1}{2\hbar}(V+iW)x^{2}}~. \label{psixy}%
\end{equation}
where $V$ and $W$ are real symmetric $n\times n$ matrices with $V>0$ (i.e.
positive definite). Such centered Gaussians are generalizations of the usual
squeezed coherent states appearing in the physical literature; see
\cite{Birk,Littlejohn}. The function $\psi_{V,W}$ is normalized to unity:
$||\psi_{V,W}||_{L^{2}}=1$ and its Wigner transform is given by
\begin{equation}
W\psi_{V,W}(z)=\left(  \tfrac{1}{\pi\hbar}\right)  ^{n}e^{-\tfrac{1}{\hbar
}Gz\cdot z} \label{wxy}%
\end{equation}
where $G$ is the symmetric and symplectic matrix%
\begin{equation}
G=%
\begin{pmatrix}
V+WV^{-1}W & WV^{-1}\\
V^{-1}W & V^{-1}%
\end{pmatrix}
. \label{gsym}%
\end{equation}
That $G\in\operatorname*{Sp}(n)$ easily follows from the observation that
$G=S^{T}S$ \ where
\begin{equation}
S=%
\begin{pmatrix}
V^{1/2} & 0\\
V^{-1/2}W & V^{-1/2}%
\end{pmatrix}
\label{bi}%
\end{equation}
clearly is symplectic. Let $\operatorname*{Gauss}_{0}(n)$ be the set of all
centered Gaussians (\ref{psixy}); a central (and often implicitly used) result
is that the metaplectic group $\operatorname*{Mp}(n)$ acts transitively on
$\operatorname*{Gauss}_{0}(n)$: we have an action%
\begin{equation}
\operatorname*{Mp}(n)\times\operatorname*{Gauss}\nolimits_{0}%
(n)\longrightarrow\operatorname*{Gauss}\nolimits_{0}(n). \label{action}%
\end{equation}

\subsection{The Radon transform of $\psi_{V,W}$}

Let us calculate the Radon transform%
\begin{equation}
R_{\psi_{_{U,V}}}(X,A,B)=(\det\Lambda^{-1})|\widehat{U}_{A,B}\psi
_{V,W}(\Lambda^{-1}X)|^{2} \label{psiuv}%
\end{equation}
of $\psi_{V,W}$ using formula (\ref{RadonDefine}). For this we have to
determine $\widehat{U}_{A,B}\psi_{V,W}$ where $\widehat{U}_{A,B}%
\in\operatorname*{Mp}(n)$ covers the symplectic rotation
\begin{equation}
U_{A,B}=%
\begin{pmatrix}
\Lambda^{-1}A & \Lambda^{-1}B\\
-\Lambda^{-1}B & \Lambda^{-1}A
\end{pmatrix}
.
\end{equation}
The most natural (and easiest) way to determine $\widehat{U}_{A,B}\psi_{V,W}$
is to use symplectic covariance formula
\[
W(\widehat{U}_{A,B}\psi_{V,W})(z)=W\psi_{V,W}(U_{A,B}^{-1}z)
\]
of the Wigner transform; it yields, taking (\ref{wxy}) into account and using
the relation $U_{A,B}^{-1}=U_{A,B}^{T}$
\[
W(\widehat{U}_{A,B}\psi_{V,W})(z)=\left(  \tfrac{1}{\pi\hbar}\right)
^{n}e^{-\tfrac{1}{\hbar}U_{A,B}GU_{^{AB}}^{T}z\cdot z}.
\]
To explicitly determine $G^{\prime}=U_{A,B}GU_{^{AB}}^{T}$ is a rather lengthy
(though straightforward) calculation; however since we have an action
(\ref{action}) we know (by (\ref{gsym})) that there will exist $V^{\prime}$
and $W^{\prime}$ such that
\[
G^{\prime}=%
\begin{pmatrix}
V^{\prime}+W^{\prime}V^{\prime-1}W^{\prime} & W^{\prime}V^{\prime-1}\\
V^{\prime-1}W^{\prime} & V^{\prime-1}%
\end{pmatrix}
\]
corresponding to the Gaussian
\[
\widehat{U}_{A,B}\psi_{V,W}=\psi_{V^{\prime},W^{\prime}}(x)=\left(  \tfrac
{1}{\pi\hbar}\right)  ^{n/4}(\det V^{\prime})^{1/4}e^{-\tfrac{1}{2\hbar
}(V^{\prime}+iW^{\prime})x^{2}}.
\]
Now, since it is only the (squared) modulus of $\widehat{U}_{A,B}\psi_{V,W}$
which appears in (\ref{psiuv}) we will have%
\[
|\widehat{U}_{A,B}\psi_{V,W}|^{2}=\left(  \tfrac{1}{\pi\hbar}\right)
^{n/2}(\det V^{\prime})^{1/2}e^{-\tfrac{1}{2\hbar}V^{\prime}x^{2}}%
\]
so that it suffices to determine $V^{\prime}$, whose inverse is the lower
right block of $G^{\prime}$. This is easily done using the relation
$G^{\prime}=U_{A,B}GU_{^{AB}}^{T}$ and one finds, after a few calculations and
simplifications,%
\begin{equation}
V^{\prime}=\Lambda\left[  BVB^{T}+(A-BW)V^{-1}(A-BW)^{T}\right]  ^{-1}\Lambda.
\label{V1}%
\end{equation}
Summarizing, we have, after insertion in (\ref{psiuv}),
\begin{align}
R_{\psi_{_{V,W}}}(X,A,B)  &  =C\label{RG0}\\
&  \times\exp\left[  \left(  -\tfrac{1}{2\hbar}\left[  BVB^{T}+(A-BW)V^{-1}%
(A-BW)^{T}\right]  ^{-1}\right)  X\cdot X\right] \nonumber
\end{align}%
\begin{equation}
C=\left(  \tfrac{1}{\pi\hbar}\right)  ^{n/2}(\det V^{\prime})^{1/2}\det\left[
BVB^{T}+(A-BW)V^{-1}(A-BW)^{T}\right]  ^{-1}. \label{RG1}%
\end{equation}
Suppose for instance $\psi_{_{V,W}}$ is a \textquotedblleft squeezed coherent
state\textquotedblright; then $W=0$ and $\psi_{V}=\psi_{V,0}$ is
\[
\psi_{V}(x)=\left(  \tfrac{1}{\pi\hbar}\right)  ^{n/4}(\det V)^{1/4}%
e^{-\tfrac{1}{2\hbar}Vx^{2}}.
\]
Its Radon transform is then
\begin{gather}
R_{\psi_{V}}(X,A,B)=C\exp\left[  -\tfrac{1}{2\hbar}(BVB^{T}+AV^{-1}A^{T}%
)^{-1}X\cdot X\right] \\
C=\left(  \tfrac{1}{\pi\hbar}\right)  ^{n/2}(\det V^{\prime})^{1/2}%
\det(BVB^{T}+AV^{-1}A^{T})^{-1}\\
V^{\prime}=\Lambda\left[  BVB^{T}+AV^{-1}A^{T}\right]  ^{-1}\Lambda
\end{gather}

\subsection{Application: the Pauli problem}

When $n=1$ the formula (\ref{RG0}) takes the simple form (replacing $A,B,V,W$
with scalars $a,b,v,w$)%
\begin{equation}
R_{\psi_{_{v,w}}}(X,a,b)=C\exp\left[  -\tfrac{1}{2\hbar}\left[  b^{2}%
v+(a-bw)^{2}v^{-1}\right]  ^{-1}X\cdot X\right]  \label{RG}%
\end{equation}
where
\[
\psi_{_{v,w}}(x)=\left(  \tfrac{1}{\pi\hbar}\right)  ^{1/4}v^{1/4}%
e^{-\tfrac{1}{2\hbar}(v+iw)x^{2}}.
\]
Let us apply this formula to the \textquotedblleft Pauli
problem\textquotedblright. This problem goes back to the question Pauli asked
in \cite{Pauli}, whether the probability densities $|\psi(x)|^{2}$ and
$|\widehat{\psi}(p)|^{2}$ uniquely determine the wavefunction $\psi(x)$. The
answer is \textquotedblleft no\textquotedblright: consider in fact the
Gaussian wavepacket
\begin{equation}
\psi(x)=\left(  \tfrac{1}{2\pi\sigma_{xx}}\right)  ^{1/4}e^{-\frac{x^{2}%
}{4\sigma_{xx}}}e^{\frac{i\sigma_{xp}}{2\hbar\sigma_{xx}}x^{2}} \label{Gauss1}%
\end{equation}
whose Fourier transform of $\psi$ is given by
\begin{equation}
\widehat{\psi}(p)=e^{i\gamma}\left(  \tfrac{1}{2\pi\sigma_{pp}}\right)
^{1/4}e^{-\frac{p^{2}}{4\sigma_{pp}}}e^{-\frac{i\sigma_{xp}}{2\hbar\sigma
_{pp}}p^{2}} \label{FGauss1}%
\end{equation}
where $\gamma$ is an unessential constant real phase. Thus,
\begin{equation}
|\psi(x)|^{2}=\left(  \tfrac{1}{2\pi\sigma_{xx}}\right)  ^{1/2}e^{-\frac
{x^{2}}{2\sigma_{xx}}}\text{ \ , \ }|\widehat{\psi}(p)|^{2}=\left(  \tfrac
{1}{2\pi\sigma_{pp}}\right)  ^{1/2}e^{-\frac{p^{2}}{2\sigma_{pp}}} \label{mod}%
\end{equation}
and these relations imply the knowledge of $\sigma_{xx}$ and of $\sigma_{pp}$,
but not of the covariance $\sigma_{xp}$ (the latter can actually be determined
up to a sign using the fact that $\psi$ saturates the
Robertson--Schr\"{o}dinger uncertainty principle: we have $\sigma_{xx}%
\sigma_{pp}-\sigma_{xp}^{2}=\tfrac{1}{4}\hbar^{2}$). Let us calculate the
Radon transform of $\psi$ using formula (\ref{RG}). We have here
$v=\hbar/2\sigma_{xx}$, $w=-\sigma_{xp}/\sigma_{xx}$ hence (\ref{RG}) becomes
\begin{equation}
R_{\psi_{_{v,w}}}(X,a,b)=C\exp\left(  -\tfrac{1}{2\hbar}\left[  \frac
{b^{2}\hbar}{2\sigma_{xx}}+\left(  a+b\frac{\sigma_{xp}}{\sigma_{xx}}\right)
^{2}\frac{2\sigma_{xx}}{\hbar}\right]  ^{-1}\right)  X^{2}. \label{RadPauli}%
\end{equation}
Notice that, as expected,
\[
R_{\psi_{_{v,w}}}(X,1,0)=\left(  \tfrac{1}{2\pi\sigma_{xx}}\right)
^{1/2}e^{-\frac{X^{2}}{4\sigma_{xx}}}=|\psi(x)|^{2}%
\]
and, using the relation $\sigma_{xx}\sigma_{pp}-\sigma_{xp}^{2}=\tfrac{1}%
{4}\hbar^{2}$,
\[
R_{\psi_{_{v,w}}}(X,0,1)=\left(  \tfrac{1}{2\pi\sigma_{pp}}\right)
^{1/2}e^{-\frac{X^{2}}{2\sigma_{pp}}}=|\widehat{\psi}(X)|^{2}.
\]
These relations show why we cannot recover the state $\psi$ using the two
radon transforms $R_{\psi_{_{v,w}}}(X,1,0)$ and $R_{\psi_{_{v,w}}}(X,0,1)$:
none of them allow to determine the covariance $\sigma_{xp}$. However, it
suffices with one clever choice of the parameters $a$ and $b$ in formula
(\ref{RadPauli}). Suppose indeed we have measured, for some values of the
parameters $a$ and $b$, the positive quantity%
\begin{equation}
K=\frac{b^{2}\hbar}{2\sigma_{xx}}+\left(  a+b\frac{\sigma_{xp}}{\sigma_{xx}%
}\right)  ^{2}\frac{2\sigma_{xx}}{\hbar} \label{K}%
\end{equation}
appearing in the exponent of (\ref{RadPauli}). Assuming that the variances
$\sigma_{xx}$ and $\sigma_{pp}$ are known, we can find the covariance
$\sigma_{xp}$ as follows: viewing (\ref{K}) as a quadratic equation in the
unknown $\sigma_{xp}$ we demand that this equation has exactly \emph{one} real
root. This imposes the relation%
\[
K=\frac{b^{2}\hbar}{2\sigma_{xx}}%
\]
and reduces (\ref{K}) to the equation%
\[
a+b\frac{\sigma_{xp}}{\sigma_{xx}}=0
\]
from which $\sigma_{xp}$ is unambiguously determined.

\section{Concluding Remarks}

In this paper we outlined a novel approach to the symplectic Radon transform,
which we believe is conceptually very simple once one realized the fundamental
role played in quantum mechanics by the metaplectic representation of the
symplectic group. As we have discussed elsewhere \cite{gohi} a few years ago ,
this is the shortest bridge between classical (Hamiltonian) mechanics, and its
refinement, quantum mechanics. This being said, our approach is somewhat
sketchy since we haven't characterized the classes of functions (or states) to
which we can apply the radon transform, limiting ourselves, for simplicity, to
the square integrable case. It is however well-known (at least by people
belonging to the harmonic analysis community), there are function spaces
invariant under metaplectic transformations which are larger than the space of
square integrable functions. We are, among other possibilities, thinking about
Feichtinger's modulation spaces \cite{Gro}, which are a very flexible tool for
creating quantum states, and which can be used together with Shubin's
pseudodifferential calculus \cite{sh87} to extend the theory (one could, for
instance, envisage the reconstruction of general observables along these lines).

We will definitely come back to these topics in a near future.

\section*{APPENDIX A: The metaplectic group $\operatorname*{Mp}(n)$}

For a detailed study of the metaplectic group $\operatorname*{Mp}(n)$ see
\cite{Birk,Leray}. For a rather \textquotedblleft soft\textquotedblright\ (but
still rigorous) approach see \cite{Littlejohn}.

Let $S=%
\begin{pmatrix}
A & B\\
C & D
\end{pmatrix}
$ be a real $2n\times2n$ matrix, where the \textquotedblleft
blocks\textquotedblright\ $A,B,C,D$ are $n\times n$ matrices. Let $J=%
\begin{pmatrix}
0 & I\\
-I & 0
\end{pmatrix}
$ the standard symplectic matrix. We have $S\in\operatorname*{Sp}(n)$ (the
symplectic group) if and only $SJS^{T}=S^{T}JS=J$. These relations are
equivalent to any of the two sets of conditions
\begin{align}
A^{T}C\text{, }B^{T}D\text{ \ \textit{are symmetric, and} }A^{T}D-C^{T}B  &
=I\label{cond12}\\
AB^{T}\text{, }CD^{T}\text{ \ \textit{are\ symmetric, and} }AD^{T}-BC^{T}  &
=I\text{.} \label{cond22}%
\end{align}
One says that $S$ is a \emph{free symplectic matrix }if $B$ is invertible,
i.e. $\det B\neq0$. To a free symplectic matrix is associated a generating
function: it is the quadratic form%
\begin{equation}
\mathcal{A}(x,x^{\prime})=\frac{1}{2}DB^{-1}x\cdot x-B^{-1}x\cdot x^{\prime
}+\frac{1}{2}B^{-1}Ax^{\prime}\cdot x^{\prime}. \label{wfree}%
\end{equation}
The terminology comes from the fact that the knowledge of $\mathcal{A}%
(x,x^{\prime})$ uniquely determines the free symplectic matrix $S$: we have%
\[%
\begin{pmatrix}
x\\
p
\end{pmatrix}
=%
\begin{pmatrix}
A & B\\
C & D
\end{pmatrix}%
\begin{pmatrix}
x^{\prime}\\
p^{\prime}%
\end{pmatrix}
\Longleftrightarrow\left\{
\begin{array}
[c]{c}%
p=\nabla_{x}\mathcal{A}(x,x^{\prime})\\
p^{\prime}=-\nabla_{x^{\prime}}\mathcal{A}(x,x^{\prime})
\end{array}
\right.
\]
as can be verified by a direct calculation (the quadratic form $\mathcal{A}$
is thus a generating function of the free symplectic matrix $S$) .

Now, to every free symplectic matrix $S_{\mathcal{A}}$ we associate two
operators $\widehat{S}_{\mathcal{A},m}$ by the formula%
\begin{equation}
\widehat{S}_{\mathcal{A},m}\psi(x)=\left(  \tfrac{1}{2\pi\hbar}\right)
^{n/2}i^{m-n/2}\sqrt{|\det B^{-1}|}\int e^{\frac{i}{\hbar}\mathcal{A}%
(x,x^{\prime})}\psi(x^{\prime})d^{n}x^{\prime} \label{qft1}%
\end{equation}
where $m$ (\textquotedblleft Maslov index\textquotedblright\ \cite{Birk})
corresponds to a choice of argument for $\det B^{-1}$: $m=0$
$\operatorname{mod}2$ if $\det B^{-1}>0$ and $m=1$ $\operatorname{mod}2$ if
$\det B^{-1}<0$ ($m$ is defined modulo $4$). It is not difficult to prove that
the generalized Fourier transforms $\widehat{S}_{\mathcal{A},m}$ are unitary
operators on $L^{2}(\mathbb{R}^{n})$. These operators generate a group, the
metaplectic group $\operatorname*{Mp}(n)$. One shows that every $\widehat{S}%
\in\operatorname*{Mp}(n)$ can be written (non uniquely) as a product
$\widehat{S}_{\mathcal{A},m}\widehat{S}_{\mathcal{A}^{\prime},m^{\prime}}$.
This group is a double covering of $\operatorname*{Sp}(n)$, the covering
projection being defined by
\begin{equation}
\pi_{\operatorname*{Mp}}:\operatorname*{Mp}(n)\longrightarrow
\operatorname*{Sp}(n)\text{ \ , \ }\pi_{\operatorname*{Mp}}(\widehat{S}%
_{\mathcal{A},m})=S_{\mathcal{A}}. \label{pimp}%
\end{equation}

\section*{APPENDIX B: The Lagrangian Grassmannian}

A linear subspace $\ell$ of the symplectic space $(\mathbb{R}^{2n},\sigma)$ is
called a Lagrangian subspace (or plane) if it is maximally isotropic for the
skew orthogonality relation $\sigma(z,z^{\prime})=0$. It must thus have
dimension $\dim\ell=n$ and $\sigma$ vanishes identically on $\ell$. The set of
all Lagrangian planes in $\mathbb{R}^{2n}$ is called the Lagrangian
Grassmannian and is denoted by $\operatorname*{Lag}(n)$. The symplectic group
$\operatorname*{Sp}(n)$ acts transitively on $\operatorname*{Lag}(n)$; in fact
the action
\begin{equation}
\operatorname*{Sp}(n)\times\operatorname*{Lag}(n)\longrightarrow
\operatorname*{Lag}(n) \label{spact}%
\end{equation}
thus defined induces a transitive action
\begin{equation}
U(n)\times\operatorname*{Lag}(n)\longrightarrow\operatorname*{Lag}(n)
\label{unact}%
\end{equation}
where $U(n)$ is the image in $\operatorname*{Sp}(n)$ of the unitary group
$U(n,\mathbb{C})$ by the monomorphism $\iota:u\longmapsto U$ defined by
$U(x,p)=(x^{\prime},p^{\prime})$ if $u(x+ip)=x^{\prime}+ip^{\prime}.$ In
matrix notation%
\[
\iota(A+iB)=%
\begin{pmatrix}
A & B\\
-B & A
\end{pmatrix}
.
\]
The conditions (\ref{cond12}), (\ref{cond22}) are here equivalent to \
\begin{align}
A^{T}B\text{ \textit{is symmetric, and} }A^{T}A+B^{T}B  &  =I\label{U1}\\
AB^{T}\text{ \textit{is\ symmetric, and} }AA^{T}+BB^{T}  &  =I\text{.}
\label{U2}%
\end{align}

The elements of $U(n)$ are the symplectic rotations of\emph{\ }$(\mathbb{R}%
^{2n},\omega)$:
\begin{equation}
U(n)=\operatorname*{Sp}(n)\cap O(2n,\mathbb{R}) \label{symprot}%
\end{equation}
and the transitivity of the action (\ref{unact}) follows: let $(\ell
,\ell^{\prime})\in\operatorname*{Lag}(n)\times\operatorname*{Lag}(n)$ and
choose two bases $\{e_{1},...,e_{n}\}$ and $\{e_{1}^{\prime},...,e_{n}%
^{\prime}\}$ of $\ell$ and $\ell^{\prime}$, respectively. Then $\{e_{1}%
,...,e_{n};Je_{1},...,Je_{n}\}$ and $\{e_{1}^{\prime},...,e_{n}^{\prime};$
$Je_{1}^{\prime},...,Je_{n}^{\prime}\}$ are both symplectic and orthogonal
bases of $(\mathbb{R}^{2n},\sigma)$. The automorphism $U$ of $\mathbb{R}^{2n}$
taking the first basis to the second is thus in $\operatorname*{Sp}(n)\cap
O(2n,\mathbb{R})$ and we have $\ell^{\prime}=U\ell$.

\begin{lemma}
Let $\ell\in\operatorname*{Lag}(n).$ There exist real $n\times n$ matrices
$A,B$ with $\operatorname*{rank}(A,B)=n$ and $A^{T}B=B^{T}A$ \ and
$AB^{T}=BA^{T}$ such that%
\[
\ell=\{(x,p)\in\mathbb{R}^{2n}:Ax+Bp=0\}.
\]

\end{lemma}

\begin{proof}
It is a straightforward consequence of the transitivity of the action
(\ref{unact}) of $U(n)$ on $\operatorname*{Lag}(n)$
\end{proof}

\begin{acknowledgement}
This work has been financed by the Grant P 33447 N of the Austrian Research
Foundation FWF.
\end{acknowledgement}

\end{document}